\documentclass[journal,onecolumn]{IEEEtran}

\usepackage{amssymb}
\usepackage{setspace}

\usepackage[shortlabels]{enumitem}
\usepackage{xcolor}
\usepackage{amssymb}
\usepackage{float}
\usepackage{amsmath}
\usepackage{amsthm}
\usepackage{tikz}
\usetikzlibrary{positioning}
\usetikzlibrary{shapes}
\usetikzlibrary{crypto.symbols}
\tikzset{shadows=no} 
\usepackage{hyperref}
\usepackage{arydshln}
\usepackage[T1]{fontenc}
\newtheorem{theorem}{Theorem}[section]
\newtheorem{lemma}[theorem]{Lemma}
\newtheorem{corollary}[theorem]{Corollary}
\newtheorem{proposition}[theorem]{Proposition}

\theoremstyle{definition}
\newtheorem{definition}[theorem]{Definition}

\theoremstyle{remark}
\newtheorem{remark}[theorem]{Remark}

\newcommand{\F}{\mathbb F}
\newcommand{\Z}{\mathbb Z}
\DeclareMathOperator{\DDT}{DDT}

\DeclareMathOperator{\Aut}{Aut}
\DeclareMathOperator{\Gal}{Gal}
\DeclareMathOperator{\charac}{char}
\begin{document}
\title{Differential biases, $c$-differential uniformity, and their relation to differential attacks
}

\author{Daniele Bartoli, 
        Lukas K\"olsch, 
        and~Giacomo Micheli
\thanks{D. Bartoli is with University of Perugia. Email: daniele.bartoli@unipg.it}
\thanks{L. K\"olsch and G. Micheli are with University of South Florida and the Center for Cryptographic Research at USF. Email: \{koelsch,gmicheli\}@usf.edu}}

\maketitle              
\begin{abstract}
Differential cryptanalysis famously uses statistical biases in the propagation of differences in a block cipher to attack the cipher. In this paper, we investigate the existence of more general statistical biases in the differences. 
To this end, we discuss the $c$-differential uniformity of S-boxes, which is a concept that was recently introduced in Ellingsen et. al. [IEEE Transactions on Information Theory, vol. 66, no. 9 (2020)] to measure certain statistical biases that could potentially be used in attacks similar to differential attacks. Firstly, we prove that a large class of potential candidates for S-boxes necessarily has large $c$-differential uniformity for all but at most $B$ choices of $c$, where $B$ is a constant independent of the size of the finite field $q$. This result implies that for a large class of functions, certain statistical differential biases are inevitable. 

In a second part, we discuss the practical possibility of designing a differential attack based on weaknesses of S-boxes related to their $c$-differential uniformity.

\end{abstract}


\section{Introduction}

\subsection{Background on block ciphers and differential cryptanalysis}

\paragraph{Block ciphers}
A block cipher is a symmetric encryption scheme that transforms an $n$-bit plaintext into an $n$-bit ciphertext using a secret key. Block ciphers (like AES) constitute the majority of all symmetric ciphers in use today, and are a staple of modern cryptography. A classical construction of block ciphers is iterative, meaning that the entire cipher is a sequence of (generally simple and almost identical) round functions, see Figure~\ref{fig:iterated} for a schematic illustration. A standard choice for a round function is a \emph{Substitution-Permutation Network (SPN)}, which consisting of a \emph{linear layer}, sometimes called the permutation part and a substitution layer (usually several so called S-boxes running in parallel), with a key addition in between the rounds, see Figure~\ref{fig:spn} for a conceptual visualization. 
\setlength{\belowcaptionskip}{-10pt}
\begin{figure} [b]
\centering
\begin{tikzpicture}[scale=.8]

  \tikzstyle{every node}=[transform shape];
  \tikzstyle{every node}=[node distance=1.2cm];
  
    \node (XOR-1)[XOR,scale=1.2] {};
    \node [left of=XOR-1] (p) {$m$};
	\node (f-1) [right of=XOR-1,draw,rectangle,thick] {$f$};
	\path[line] (p) edge (XOR-1);
	\path[line] (XOR-1) edge (f-1);
    
    \node (XOR-2)[right of=f-1,XOR,scale=1.2] {};
	\path[line] (f-1) edge node[below] {} (XOR-2);
	\node (f-2) [right of=XOR-2] {};
	\path[line] (XOR-2) edge (f-2) ;
	\node (dots)[right of=XOR-2,node distance=1.5cm] {$\dots$};
	
    \node (XOR-3)[right of=dots,XOR,scale=1.2,node distance=1.5cm] {};
	\path[line] (dots) edge (XOR-3) ;
	\node (f-3) [right of=XOR-3,draw,rectangle,thick] {$f$};
    	\path[line] (XOR-3) edge (f-3);

    \node (XOR-4)[right of=f-3,XOR,scale=1.2] {};
	\path[line] (f-3) edge node[below] {} (XOR-4);
	\node [right of=XOR-4] (c) {$c$};
	\path[line] (XOR-4) edge (c);

	\node (k-0) [above=1.3cm of XOR-1] {};
	\path[line] (k-0) edge node[right] {\small $k_{0}$} (XOR-1);

	\node (k-1) [above=1.3cm of XOR-2] {};
	\path[line] (k-1) edge node[right] {\small $k_{1}$} (XOR-2);

	\node (k-2) [above=1.3cm of XOR-3] {};
	\path[line] (k-2) edge node[right] {\small $k_{r-1}$} (XOR-3);

	\node (k-3) [above=1.3cm of XOR-4] {};
	\path[line] (k-3) edge node[right] {\small $k_{r}$} (XOR-4);

\end{tikzpicture}
\caption{An iterated (key-alternating) block cipher with $r$ rounds and subkeys~$k_i$ that encrypts a plaintext $m$ into a ciphertext $c$}
\label{fig:iterated}
\end{figure}
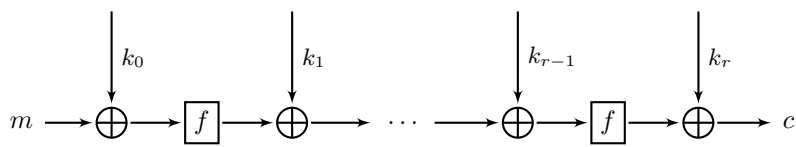

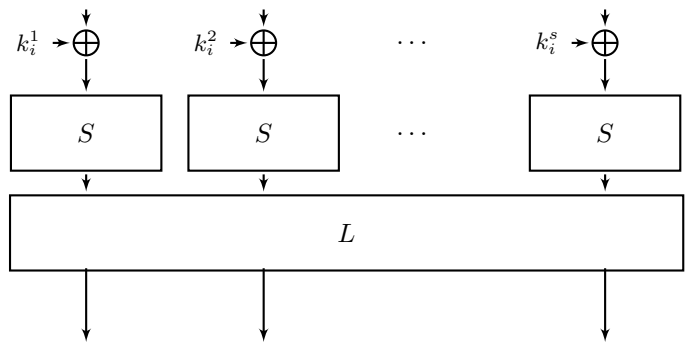
\begin{figure} 
\centering
\begin{tikzpicture}[scale=.8]

  \tikzstyle{every node}=[transform shape];
  \tikzstyle{every node}=[node distance=1.2cm];

    \node (XOR-1)[XOR,scale=1] {};
	\node (S1) [below of=XOR-1,draw,rectangle,thick,minimum width=2cm,minimum height=1cm] {$S$};
	    \node (XOR-2)[XOR,right=2cm of XOR-1,scale=1] {};
	\node (S2) [below of=XOR-2,draw,rectangle,thick,minimum width=2cm,minimum height=1cm] {$S$};
	\path[line] (XOR-2) edge (S2);
	
	\path[line] (XOR-1) edge (S1);
	\node (L) [below right=0.3cm and -4.4cm of S2,draw,rectangle,thick,minimum width=8.95cm,minimum height=1cm] {$L$};
	
		\node (dots)[right of=XOR-2,node distance=2cm] {$\dots$};
		\node (dots2)[right of=S2,node distance=2cm] {$\dots$};
	
	    \node (XOR-3)[XOR,right=2cm of dots,scale=1] {};
	\node (S3) [below of=XOR-3,draw,rectangle,thick,minimum width=2cm,minimum height=1cm] {$S$};
	\path[line] (XOR-3) edge (S3);
	
		\node (l-0) [below=0.3cm of S1] {};
	\path[line] (S1) edge node[below] {} (l-0);

	\node (l-1) [below=0.3cm of S2] {};
	\path[line] (S2) edge node[below] {} (l-1);

	\node (l-2) [below=0.3cm of S3] {};
	\path[line] (S3) edge node[below] {} (l-2);
	
		\node (o-0) [below=1cm of S1] {};
				\node (o2-0) [below=2.3cm of S1] {};
	\path[line] (o-0) edge node[below] {} (o2-0);

	\node (o-1) [below=1cm of S2] {};
					\node (o2-1) [below=2.3cm of S2] {};
	\path[line] (o-1) edge node[below] {} (o2-1);

	\node (o-2) [below=1cm of S3] {};
					\node (o2-2) [below=2.3cm of S3] {};
	\path[line] (o-2) edge node[below] {} (o2-2);
	
	\node (m-0) [above=0.3cm of XOR-1] {};
	\path[line] (m-0) edge node[below] { } (XOR-1);

	\node (m-1) [above=0.3cm of XOR-2] {};
	\path[line] (m-1) edge node[below] {} (XOR-2);

	\node (m-2) [above=0.3cm of XOR-3] {};
	\path[line] (m-2) edge node[below] {} (XOR-3);

	\node (k-0) [left=0.3cm of XOR-1] {\small $k_{i}^1$};
	\path[line] (k-0) edge node[right] {} (XOR-1);

	\node (k-1) [left=0.3cm of XOR-2] {\small $k_{i}^2$};
	\path[line] (k-1) edge node[right] {} (XOR-2);

	\node (k-2) [left=0.3cm of XOR-3] {\small $k_{i}^s$};
	\path[line] (k-2) edge node[right] {} (XOR-3);
	
\end{tikzpicture}
\caption{A high-level view of one round of an SPN with an S-box $S$, linear layer $L$ and round key $k_i$}
\label{fig:spn}
\end{figure}

\paragraph{Differential cryptanalysis}
One of the most important attacks against block ciphers is the so-called differential attack introduced in~\cite{diffattack}, which also serves as a basis for more advanced attacks like boomerang attacks~\cite{boomattack}, rectangle attacks~\cite{rectattack}, or differential-linear attacks~\cite{dlattacks}.
  
Let us briefly recall the main idea behind the classical differential attack. The attack uses that differences of plain texts propagate with different probabilities through the encryption process. 

A cipher vulnerable to a differential attack has a (strongly) non-uniform distribution of differences that can then be exploited by a differential attack. The concept of \emph{difference} used here is usually the XOR (i.e. addition in $\F_{2}^n$). The main reason for this is that the most standard design for SPNs uses \emph{key additions}, i.e. the key is simply added to the message in between all rounds. 

Clearly, we have $(x+\Delta+k)-(x+k)=(x+\Delta)-x$, so the key addition process does not impact the propagation of differences at all. Further, the differences are also not impacted by the linear layer of SPNs, so the only part that has direct influence is the S-box layer, which makes both analysis and design considerably easier. For a function $F\colon \F_{2}^n \rightarrow \F_2^n$, the differential attack thus considers the following distribution of probabilities, where $\Delta_I$, $\Delta_O$ are the differences of plain and cipher texts, respectively:
\begin{equation}
P_{\Delta_I,\Delta_O}=P_x(F(x+\Delta_I)=\Delta_O+F(x)).
\label{eq:prob}
\end{equation}

The main idea behind the differential attack is then that, because the key does not impact the propagation of differences, differences can be broken down over all rounds, so \emph{differential trails} $(\Delta_0,\dots,\Delta_r)$ can be constructed which capture the probability that an input difference $\Delta_0$ propagates through an $r$-round block cipher to an output difference $\Delta_r$ via intermediate differences $\Delta_i$ after the $i$-th round. With the tacit assumption of uniformity, the differential analysis of the cipher can thus be broken down to analysing its constitutive rounds, which, in the case of an successful attack, allows the attacker to guess the key given enough plaintext - ciphertext pairs. To combat these differential attacks, the probabilities in Eq.~\eqref{eq:prob} should be as uniform as possible.

\paragraph{Differential uniformities of S-boxes} As described above, the resistance of an SPN relies on its non-linear part, i.e. the choice of the S-box. The behaviour of an S-box $F$ with regard to differential attacks is measured by its \emph{Difference Distribution Table (DDT)} and \emph{differential uniformity} $\delta(F)$ which are defined as

\begin{align}
\DDT_F[a,b]&=\#\{x\in \F_{2^n} \colon F(x+a)+F(x)=b\}, \label{eq:ddt}  \\
	\delta_F&=\max_{a \in {\F_2^n}^*, b \in \F_2^n}\DDT_F[a,b].\label{eq:diff_uni}
\end{align}
The lower the differential uniformity of $F$, the better is its resistance against a differential attack.\footnote{Of course, the choice of the linear layer is also important for the resistance against a differential attack to maximize the number of active S-boxes.} The best possible differential uniformity for S-boxes over binary finite fields is 2, those S-boxes are called almost perfect nonlinear (APN).

\paragraph{Alternative differences and $c$-differential uniformities}
The basic idea of the differential attack can be transferred to another group operation that substitutes the role of the XOR/addition. For example, \emph{multiplicative differentials} were considered in~\cite{mult}. Instead of investigating the propagation on pairs $(x,x+\Delta)$, this kind of differential attack considers pairs $(x,\alpha \cdot x)$, where the multiplication is the multiplication in the ring $\Z_n$. These kind of attacks were motivated by a number of ciphers that did not just use a simple key addition, but instead relied on modular multiplication as a primitive operation. The authors of~\cite{cdiff} use this motivation to consider a notion that tracks some statistical behaviour of the propagation of differences and is clearly inspired by the definition of the differential uniformity in~\eqref{eq:diff_uni}. The definition uses a more general setting (not restricting to characteristic $2$), and uses the well known isomorphism (as vector spaces) of $\F_p^n$ and $\F_{p^n}$.
\begin{definition}\cite[Definition 1]{cdiff}
	Let $F:\F_{p^n}\rightarrow \F_{p^n}$, and $c\in\F_{p^n}$.
For $a,b\in\mathbb{F}_{p^n}$, we let the entries of the $c$-Difference Distribution Table ($c$-DDT) be defined by $_c\DDT_F[a,b]=\#{\{x\in\mathbb{F}_{p^n} : F(x+a)-cF(x)=b\}}$. We call the quantity
\[
_c\delta_{F}=\max\left\{{_c\DDT}_F[a,b]\,:\, a,b\in \mathbb{F}_{p^n}, \text{ and } a\neq 0 \text{ if $c=1$} \right\}\]
the {\em $c$-differential uniformity} of~$F$.

\end{definition}
In characteristic $2$, the case $c=1$ clearly recovers the usual definition of DDT and differential uniformity in~\eqref{eq:ddt} and \eqref{eq:diff_uni}. Again, the function $F$ has minimal statistical bias if the $c$-differential uniformity is as low as possible. If a function $F$ achieves the minimal possible $c$-differential uniformity $1$, we call it  \emph{perfect $c$-nonlinear (PcN)}. 
\begin{remark}
    Since our investigations in this paper are largely independent of the characteristic of the underlying finite field, following~\cite{cdiff}, we chose to discuss the $c$-differential uniformity and its related concepts in full generality in all positive characteristics.  We want to note that while symmetric cryptography in odd characteristic remains a rarity, there are some recent developments in this direction~\cite{oddchar_sym}.
\end{remark}

The new notion of $c$-differential uniformity has led to much new research, counting more than a dozen papers in less than two years dedicated to it (see e.g.~\cite{cdiff1,cdiff3,cdiff2}, and in particular the survey paper~\cite{cdiff_survey} and the references therein), usually focusing on determining the $c$-differential uniformity of functions $F \colon \F_{p^n} \rightarrow \F_{p^n}$ that have seen use or are possible candidates for S-boxes. Despite the considerable attention that the $c$-differentials have received, a practical analysis has so far been lacking.

\subsection{Our Contribution and organization of the paper}

The study of the $c$-differential uniformity serves as an example of two important and much more general questions for the design of block ciphers: 
\begin{itemize}
    \item Is it possible to give general conditions when statistical biases for functions of a certain pattern are \emph{guaranteed}?
    \item If we know that a function has a statistical bias of a certain form, can we find a framework to analyze the possibility of a practical attack using this bias?
\end{itemize}
In this paper, we are going to give answers to these two questions using the example of biases encoded by $c$-differential uniformities. 


We want to emphasize that the purposes of this paper is to show that certain differential biases are essentially inevitable, and to provide a general discussion on whether it is possible to mount attacks based on those biases. It remains an open question whether these biases (that provably exist) can be used to mount a specific attack on a cipher that is currently in use.


Broadly speaking, the contribution of this paper is thus twofold: In Section~\ref{sc:1}, we show that for a wide class of functions $F \colon \F_{p^n} \rightarrow \F_{p^n}$ the $c$-differential uniformity is actually the \emph{worst} possible for almost all $c$. More precisely,  we show in Theorem~\ref{thm:main} that in characteristic $2$ there is an \emph{explicit} constant $B$, independent of the size of the finite field $\mathbb F_q$, such that for all $c$'s outside a set of size $B$  any function of odd degree whose first and second Hasse derivative is non-zero has worst possible $c$-differential uniformity. Note that functions with second Hasse derivative equal to zero are highly non-generic and can be completely characterized (see Remark \ref{rem:secondhasse}). The proofs in this section use novel techniques for this line of research, relying on the theory of algebraic function fields, Galois theory, and algebraic geometry. 

 The results of Section~\ref{sc:1} indicate that in many cases, some extreme statistical biases in the $c$-differences are actually inevitable. This result of course leads to the obvious question whether  or how one can use those biases to mount an attack on a cipher.

In Section~\ref{sc:2}, we discuss the possibility of constructing such an attack. In particular, we investigate what kind of cipher could be theoretically vulnerable to an attack based on a statistical weakness in the $c$-differential uniformity and we relate the $c$-differential uniformity to a special kind of "regular" differential attack. To do this, we introduce a general form of differential uniformity using a (more  arbitrary) binary operation $\circ$ replacing the usual XOR. We in particular investigate the interaction of these generalized differences with the linear layer and the key addition process of a standard SPN. 

We finish our paper with a conclusion and some interesting practical and theoretical questions related to our results.

\section{On the $c$-differential uniformity of low degree polynomials} \label{sc:1}

Whereas most of the papers that appeared so far in the literature on $c$-differential uniformities deal with monomials (see e.g. \cite{Bartoli2020,cdiff,cdiff1,Mesnager,Wang,Tu,zha2021some}), in this section, we investigate necessary constraints on polynomials with low (i.e. good)  $c$-differential uniformity. We start with a basic overview over the tools we employ to get the result.

\subsection{Tools and Notation}
\paragraph{Global Function Fields and Galois Theory}
We give a brief overview on function fields and Galois theory to the extent that is needed for this section. We broadly follow the notation from~\cite{stichtenoth2009algebraic}.

Let $\mathbb F_q(t)$ be the rational function field in the variable $t$ over the finite field of order $q$. A function field $F$ is an algebraic extension of $\mathbb F_q(t)$. The field of constants $k_F$ of $F$ is the subfield of elements of $F$ that are algebraic over $\mathbb F_q$.

A valuation ring $\mathcal O$ is a subring of $F$ that contains $\mathbb F_q$, is different from $F$,  and such that if $x\not\in \mathcal O$, then $1/x\in \mathcal O$. A place $P$ is the unique maximal ideal of a valuation ring $\mathcal O$ and it is principal. The valuation ring attached whose maximal ideal is the place $P$ is denoted by $\mathcal O_P$.

Fix a place $P$. Each element $z\in F\setminus \{0\}$ can be uniquely written as $z=t^nu$, where $t$ is a prime element for a place $P$  and $u \in \mathcal O_P$ is invertible. We associate to $P$ a function $v_P : F\to \mathbb{Z}\cup \{\infty\}$, called valuation at $P$, defined as $v_P(z)=n$, if $z=t^nu \in F\setminus \{0\}$,  and $v_P(z)=\infty$, if $z=0$.

An inclusion $F\subseteq L$ of function fields is said to be a function field extension, and we will denote it by $L:F$. The degree of $L:F$ is simply the integer $[L:F]:=\dim_{F}(L)$. If $P$ is a place of $F$, there are $Q_1,\dots, Q_\ell$ places of $L$ above $P$, i.e. places of $L$ that contain $P$. The relative degree of $Q_i$ over $P$ is the integer $f(Q_i\mid P)=[\mathcal O_{Q_i}/ Q_i:\mathcal O_P/P]$. 

The ramification index $e(Q_i\mid P) := e$ of $Q_i$ over $P$ is a natural number such that 
$$v_{Q_i}(x)=e\cdot v_P(x),  \quad \forall  x \in  F.$$
We say that $Q_i\mid P$ is ramified if $e(Q_i\mid P) >1$, and unramified if $e(Q_i\mid P) =1$.
The fundamental equality for function field extensions states 
\[\sum^{\ell}_{i=1} f(Q_i\mid P)e(Q_i\mid P)=[L:F].\]
A finite separable extension of fields $M\supseteq F$ is said to be Galois if $\Aut_F(M):=\{f\in \Aut(M): \; f(x)=x \; \forall x\in F\}$ has size $[M:F]$. Let us recall that every splitting field $M$ of a separable polynomial in $F[x]$ is a Galois extension of $F$. Moreover, every Galois extension of $F$ can be seen as a splitting field of some polynomial in $F[x]$.

\begin{definition}
  Let $M:F$ be a Galois extension of function fields and let $R$ be a place of $M$ lying over a place $P$ of $F$. Then we define the decomposition group as
\[D(R \mid P)=\{g\in G\mid g(R)=R\}\]
and the inertia group as
\[I(R\mid P)=\{g\in D(R\mid P) \mid v_R(g(s)-s)\geq 1\quad \forall s\in \mathcal O_R\}. \]
\end{definition}

The following result is useful to connect the action of the Galois group on the roots to the intermediate splitting of places. See  \cite[Satz 1]{van1935zerlegungs}.
\begin{lemma}\label{lemma:orbits}
Let $L:K$ be a finite separable extension of function fields, let $M$ be its Galois
 closure and $G:= \Gal(M:K)$ be its Galois group. 
 Let $P$ be a place of $K$ and $\mathcal Q$ be the set of places of $L$ lying above $P$.
Let $R$ be a place of $M$ lying above $P$. Then we have the following:
\begin{enumerate}
\item There is a natural bijection between $\mathcal Q$ and the set of orbits of $H:=\mathrm{Hom}_K(L,M)$ under the action of the decomposition group $D(R\mid P)=\{g\in G\,\mid \, g(R)=R\}$.
\item  Let $Q\in \mathcal Q$ and let $H_Q$ be the orbit of $D(R\mid P)$ corresponding to $Q$. Then $|H_Q|=e(Q\mid P)f(Q\mid P)$ where $e(Q\mid P)$ and $f(Q\mid P)$ are ramification index and relative degree, respectively. 

\item The orbit $H_Q$ partitions further under the action of the inertia group $I(R\mid P)$ into $f(Q\mid P)$ orbits of size $e(Q\mid P)$. 

\end{enumerate}
\end{lemma}

It follows immediately that
\begin{corollary}\label{cor:nuovo}
 If $L:F$ is ramified at $P$ then $D(R\mid P)$ is non-trivial. In particular, $L: F$  is ramified at $P$ if and only if $M: F$ is ramified at $P$.
\end{corollary}

Notice that if $M:K$ is a Galois extension of function fields over $\mathbb F_q$, then $kM:kK$ is Galois for any extension $k$ of $\mathbb F_q$ (including $k=\overline{\mathbb F}_q)$.
It is well known, see for example \cite{stichtenoth2009algebraic}, that the geometric Galois group is generated by the inertia groups.
\begin{lemma}\label{lemma:genbyinertia}
Let $L:K$ be a finite separable extension of function fields, let $M$ be its Galois
 closure and let $G:= \Gal(\overline{\mathbb F}_qM:\overline{\mathbb F}_qK)$.
 Then $G$ is generated by the inertia groups 
 $I(R\mid P)$, i.e. \[G=\langle I(R\mid P):  \; P \; \text{place of $K$,}\quad R\mid P,\quad R \; \text{place of $M$} \rangle.\]
\end{lemma}

The following result can be deduced from \cite{kosters2017short}.

\begin{theorem}\cite[Theorem 3.3]{bartoli2022algebraic}
\label{thm:existence_tot_split}
Let $p$ be a prime number, $m$ a positive integer, and $q=p^m$.
Let $L:F$ be a separable extension of global function fields over $\mathbb F_q$ of degree $n$,  $M$ be the Galois closure of $L:F$, and suppose that the field of constants of $M$ is $\mathbb F_q$.
There exists an explicit constant $C\in \mathbb R^+$ depending only on the genus of $M$ and the degree of $L:F$ such that if $q>C$ then $L:F$ has a totally split place.
\end{theorem}

The following corollary shows that we can deduce arithmetic information (splitting over $\mathbb F_q$) from geometric information (splitting over $\overline{\mathbb F}_q$, i.e. ramification).
\begin{corollary}\label{cor:geomtoarith}
Let $p$ be a prime number, $m$ a positive integer, and $q=p^m$.
Let $L:F$ be a separable extension of global function fields over $\mathbb F_q$ of degree $n$,  $M$ be the Galois closure of $L:F$. Then if $\Gal(\overline{\mathbb F}_qM:\overline{\mathbb F}_qF)=S_n$, we have that $L:F$ has a totally split place (over $\mathbb F_q$).
\end{corollary}
\begin{proof}
It is a standard fact that the Galois group of $\overline{\mathbb F}_qM:\overline{\mathbb F}_qF$ is equal to  the Galois group of $M:k_F F$, where $k_F$ is the constant field of $M$. Therefore, since
$S_n=\Gal(M:k_F F)\leq \Gal(M: F)\leq S_n$, we have that $k_F=\mathbb F_q$ by the Galois correspondence, and the final claim follows using Theorem \ref{thm:existence_tot_split}.

\end{proof}

The following result follows from \cite[Proposition 1, Page 275]{hashimoto2003galois}.
\begin{proposition}\label{prop:gensn}
Let $G$ be a transitive subgroup of $S_n$ generated by transpositions. Then $G=S_n$.
\end{proposition}

\paragraph{Curves and Varieties over Finite Fields}

As a notation, $\mathbb{P}^r(\mathbb{F}_q)$ and $\mathbb{A}^r(\mathbb{F}_q)$ denote the projective and the affine space of dimension $r\in \mathbb{N}$ over the finite field $\mathbb{F}_q$, $q$ a prime power. In the cases $r=2$ and $r=3$ we denote by $\ell_{\infty}:=\mathbb{P}^2(\mathbb{F}_q)\setminus \mathbb{A}^2(\mathbb{F}_q)$ and $H_{\infty}:=\mathbb{P}^3(\mathbb{F}_q)\setminus \mathbb{A}^3(\mathbb{F}_q)$ the line and the plane at infinity respectively.  A variety and more specifically a curve, i.e. a variety of dimension 1, is described by a certain set of equations with coefficients  in a finite field $\mathbb{F}_q$. We say that a variety $\mathcal{V}$ is \emph{absolutely irreducible} if there are no varieties $\mathcal{V}^{\prime}$ and $\mathcal{V}^{\prime\prime}$ defined over the algebraic closure of $\mathbb{F}_q$ and different from $\mathcal{V}$ such that $\mathcal{V}= \mathcal{V}^{\prime} \cup \mathcal{V}^{\prime\prime}$.
If a variety $\mathcal{V}\subset \mathbb{P}^r(\mathbb{F}_q)$ is defined by homogeneous polynomials $F_i(X_0,\ldots, X_r)=0$, for  $i=1,\ldots, s$, an $\mathbb{F}_{q}$-rational point of $\mathcal{V}$ is a point $(x_0:\ldots:x_r) \in \mathbb{P}^r(\mathbb{F}_q)$ such that $F_i(x_0,\ldots, x_r)=0$, for  $i=1,\ldots s$. A point is affine if $x_0\neq 0$. The set of the $\mathbb{F}_q$-rational points of $\mathcal{V}$ is usually denoted by $\mathcal{V}(\mathbb{F}_q)$. We  denote by the same symbol homogenized polynomials and their dehomogenizations, if the context is clear. For a more comprehensive introduction to algebraic varieties and curves we refer to \cite{HKT,Hartshorne}.

In this paper we will mostly make use of hypersurfaces, i.e varieties in $\mathbb{P}^r(\mathbb{F}_q)$ of dimension $r-1$, and specifically curves ($r=2$) and surfaces  ($r=3$). Any hypersurface is defined by a unique homogeneous $f(X_0,\ldots,X_r)$ polynomial in $r+1$ variables. For the sake of convenience, for a hypersurface $\mathcal{W}: f(X_0,\ldots,X_r)=0$ we will  also make use of its affine equation $f(1,X_1\ldots,X_r)=0$. 

Singular points of algebraic curves and surfaces can be investigated via the so-called Hasse derivative; see also \cite[Page 148]{HKT}.
\begin{definition}[\cite{Hasse1936}]
 Let $F(X_1,\ldots,X_r)\in \mathbb{F}_q[X_1,\ldots,X_r]$ be a polynomial. For any $\alpha_1,\ldots,\alpha_r\in \overline{\mathbb{F}}_q$, $F(X_1+\alpha_1,\ldots,X_r+\alpha_r)$ can be written
uniquely as
$$F(X_1+\alpha_1,\ldots,X_r+\alpha_r)=\sum_{(i_1,\ldots,i_r) \in \mathbb{N}^r}F^{(i_1,\ldots,i_r)}(\alpha_1,\ldots,\alpha_r)X_1^{i_1}\cdots X_r^{i_r},$$ 
for some polynomials $F^{(i_1,\ldots,i_r)}(X_1,\ldots,X_r) \in \mathbb{F}_q[X_1,\ldots,X_r]$. For a given multi-index $i=(i_1,\ldots,i_r) \in \mathbb{N}^r$, we define
the $i$-th Hasse derivative of $F(X_1,\ldots,X_r)$ as the polynomial $F^{(i_1,\ldots,i_r)}(X_1,\ldots,X_r) \in \mathbb{F}_q[X_1,\ldots,X_r]$.
\end{definition}

It can be seen that for any monomial $X_1^{j_1}\cdots X_r^{j_r}$ its $i=(i_1,\ldots,i_r)$-th Hasse derivative is $$\binom{i_1}{j_1}\cdots \binom{i_r}{j_r}X_1^{j_1-i_1}\cdots X_r^{j_r-i_r}$$
and vanishes if $i_k>j_k$ for some $k$; see also \cite{Goldschmidt}.

As a notation, if a polynomial $f$ is univariate, we denote its derivatives as $f^{\prime}$, $f^{\prime\prime}$, \ldots.   

Let $F(X,Y)\in \mathbb{F}_q[X,Y]$ be a polynomial defining an affine plane curve $\mathcal{C}$, let $P=(u,v)\in \mathbb{A}^2(\mathbb{F}_q)$ be a point in the plane, and write
\[
F(X+u,Y+v)=F_0(X,Y)+F_1(X,Y)+F_2(X,Y)+\cdots,
\]
where $F_i$ is either zero or homogeneous of degree $i$.

The \emph{multiplicity} of $P\in \mathcal{C}$, written as $m_P(\mathcal{C})$, is the smallest integer $m$ such that $F_m\ne 0$ and $F_i=0$ for $i<m$; the polynomial $F_m$ is the \emph{tangent cone} of $\mathcal{C}$ at $P$. A linear divisor of the tangent cone is called a \emph{tangent} of $\mathcal{C}$ at $P$. The point $P$ is on the curve $\mathcal{C}$ if and only if $m_P(\mathcal{C})\ge 1$. If $P$ is on $\mathcal{C}$, then $P$ is a \emph{simple} point of $\mathcal{C}$ if $m_P(\mathcal{C})=1$, otherwise $P$ is a \emph{singular} point of $\mathcal{C}$. 
A quick criterion to decide whether an affine  point $P$ is singular is the following: $P$ is singular if and only if $F(P)=F^{(1,0)}(P)=F^{(0,1)}(P)=0$. 

It is possible to define in a similar way the multiplicity of an ideal point of $\mathcal{C}$, that is a point of the curve lying on the line at infinity.

Given two plane curves $\mathcal{A}$ and $\mathcal{B}$ and a point $P$ on the plane, the \emph{intersection number} $I(P, \mathcal{A} \cap \mathcal{B})$ of $\mathcal{A}$ and $\mathcal{B}$ at the point $P$ can be defined by seven axioms. We do not include its precise and long definition here. For more details, we refer to \cite{Fulton} and \cite{HKT} where the intersection number is defined equivalently in terms of local rings and in terms of resultants, respectively.

Concerning the intersection number, the following two classical results  can be found in most of the textbooks on algebraic curves.
\begin{lemma}\label{le:ordinary_singular}
	Let $\mathcal{A}$ and $\mathcal{B}$ be two plane curves.  For any affine point $P$, the intersection number satisfies the inequality
	\[ I(P, \mathcal{A}\cap \mathcal{B})\ge m_P(\mathcal{A}) m_P(\mathcal{B}), \]
	with equality if and only if the tangents at $P$ to $\mathcal{A}$ are all distinct from the tangents at $P$ to $\mathcal{B}$.
\end{lemma}

\begin{theorem}[B\'ezout's Theorem]\label{th:bezout}
	Let $\mathcal{A}$ and $\mathcal{B}$ be two projective plane curves over an algebraically closed field $\mathbb{K}$, having no component in common. Let $A$ and $B$ be the polynomials associated with $\mathcal{A}$ and $\mathcal{B}$ respectively. Then
	\[
	\sum_P I(P, \mathcal{A}\cap \mathcal{B})=\deg A \cdot  \deg B,
	\]
	where the sum runs over all points in the projective plane $\mathbb{P}^2(\mathbb{K})$.
\end{theorem}

Concerning surfaces in $\mathbb{P}^3(\mathbb{F}_q)$, i.e. variety of dimension 2 defined by a homogeneous polynomial $F(X,Y,Z,T)\in \mathbb{F}_q[X,Y,Z,T]$, the same definitions for singular points and multiplicity hold. In particular a point $P$ is singular for a surface $\mathcal{S}:F(X,Y,Z,T)=0$ if and only if 
$$F(P)=F^{(1,0,0,0)}(P)=F^{(0,1,0,0)}(P)=F^{(0,0,1,0)}(P)=F^{(0,0,0,1)}(P)=0.$$

The following is a simple result about the irreducibility of plane sections of absolutely irreducible surfaces. We include here its proof for the sake of clarity.
\begin{proposition}\label{Prop:Sections}
Let $\mathcal{S}\subset \mathbb{P}^{3}(\mathbb{F}_q)$ be an absolutely irreducible surface and consider a plane $\pi$. A singular point $O$ for the curve $\pi \cap \mathcal{S}$ is either a singular point for $\mathcal{S}$ or $\pi$  is the tangent plane to $\mathcal{S}$ at $O$. In particular, if $\pi \cap \mathcal{S}$ is reducible then either $\pi$ is the tangent plane at some point of $\pi \cap \mathcal{S}$ or $\pi$ contains a singular point of $\mathcal{S}$. 
\end{proposition}
\begin{proof}
Without loss of generality we can suppose that $O$ is the origin and $\pi: Z=0$ and that $\mathcal{C}:=\pi \cap \mathcal{S}: F(X,Y)=0$,  for some  polynomial $F\in \mathbb{F}_q[X,Y]$ with $F(0,0)=F^{(1,0)}(0,0)=F^{(0,1)}(0,0)=0$. This means that the affine equation of $\mathcal{S}$ is of type $F(X,Y)+ZH(X,Y,Z)=0$ for some $H(X,Y,Z)\in \mathbb{F}_q[X,Y,Z]$. Now, either there exists a constant term in $H(X,Y,Z)$ and thus $O$ is nonsingular for $\mathcal{S}$ and $\pi$ is the tangent plane at $O$ to $\mathcal{S}$ or $H(X,Y,Z)$ possesses monomials of degree at least one and thus $O$ is a singular point for $\mathcal{S}$. 

The second part of the claim directly follows, observing that if $\mathcal{C}:=\pi \cap \mathcal{S}$ is reducible then the curve $\mathcal{C}$ possesses singular points (possibly defined over $\overline{\mathbb{F}}_q$). 

\end{proof}

Finally, we include here references for estimates on the number of $\mathbb{F}_q$-rational points of algebraic varieties over finite fields. The most celebrated result is the Hasse-Weil Theorem.

\begin{theorem}\label{Th:HW}[Hasse-Weil bound for curves]
Let $\mathcal{C}\subset \mathbb{P}^n(\mathbb{F}_q)$ be a projective absolutely irreducible non-singular curve of genus $g$ defined over $\mathbb{F}_q$. Then
\begin{equation}\label{EQ:HW1}
q+1-2g\sqrt{q}\leq \#\mathcal{C}(\mathbb{F}_q)\leq q+1+2g\sqrt{q}.
\end{equation}
\end{theorem}

\noindent If $\mathcal{C}$ is a non-singular \emph{plane} curve, then $g=(d-1)(d-2)/2$, where $d$ is the degree of the curve $\mathcal{C}$, and \eqref{EQ:HW1} reads 
\begin{equation}\label{EQ:HW2}
q+1-(d-1)(d-2)\sqrt{q}\leq \#\mathcal{C}(\mathbb{F}_q)\leq q+1+(d-1)(d-2)\sqrt{q}.
\end{equation}

If the curve $\mathcal{C}$ is singular, there is some ambiguity in defining what an $\mathbb{F}_q$-rational point of $\mathcal{C}$ actually is.  Clearly, if $\mathcal{C}$ is non-singular, then there is a bijection between $\mathbb{F}_q$-rational places (or branches) of the function field associated with $\mathcal{C}$ and $\mathbb{F}_{q}$-rational points of $\mathcal{C}$. In the singular case, this is no more true.  We refer the interested readers to \cite[Section 9.6]{HKT} where other relations are investigated. We point out that actually the bound \eqref{EQ:HW2} holds even for singular (absolutely irreducible) curves; \cite[Corollary 2.5]{AubryPerret}.

Concerning algebraic varieties of dimension larger than one, the first estimate on the number of $\mathbb{F}_{q}$-rational points  was given by Lang and Weil \cite{LangWeil} in 1954.
\begin{theorem}\label{Th:LangWeil}[Lang-Weil Theorem]
Let $\mathcal{V}\subset \mathbb{P}^N(\mathbb{F}_q)$ be an absolutely irreducible variety of dimension $n$ and degree $d$. Then there exists a constant $C$ depending only on $N$, $n$, and $d$ such that 
\begin{equation}\label{Eq:LW}
\left|\#\mathcal{V}(\mathbb{F}_q)-\sum_{i=0}^{n} q^i\right|\leq (d-1)(d-2)q^{n-1/2}+Cq^{n-1}.
\end{equation}
\end{theorem}

\noindent Although the constant $C$ was not  computed in \cite{LangWeil}, explicit estimates have been provided for instance in  \cite{CafureMatera,Ghorpade_Lachaud,Ghorpade_Lachaud2,LN1983,WSchmidt,Bombieri} and they have the general shape $C=f(d)$ provided that $q>g(n,d)$, where $f$ and $g$ are polynomials of (usually) small degree. We refer to \cite{CafureMatera} for a survey on these bounds.


\subsection{Results on the $c$-differential uniformity} 

In what follows we will consider polynomials $f(x)\in \mathbb{F}_q[x]\setminus \mathbb{F}_q[x^p]$, $q=p^n$, which are not monomials. Note that the polynomials $f(x) \in \mathbb{F}_q[x^p]$ are not a good choice for an S-box because the map $x\mapsto x^p$ is linear.

The first result we show concerns the 2-transitivity of the geometric Galois group of $F=f(x+a)- cf(x)-t\in \mathbb F_q(t)[X]$.

\begin{lemma}\label{Lemma}
Let $q=p^n$, $p$ a prime, and $f\in \mathbb{F}_{q}[x]\setminus \mathbb{F}_{q}[x^p]$, with $d=\deg(f)$, $p\nmid d(d-1)$, and suppose that $f(x)$ is not a monomial. Then, the number of $c\in \mathbb{F}_q$ for which there exists $a \in \mathbb{F}_q^*$ such that the geometric Galois group of $F=f(x+a)- cf(x)-t\in \mathbb F_q(t)[X]$ is not 2-transitive is bounded by an explicit constant $B$ independent of $q$.
\end{lemma}
\begin{proof}
It is well known that the geometric Galois group of $F$ is 2-transitive if and only if the curve $(F(X)-F(Y))/(X-Y)=0$ is absolutely irreducibile; see for instance \cite[Theorem 6.11]{LMT}.

Consider the surface 
\begin{eqnarray*}
\mathcal{W}_c &:& \frac{f(X+Z)- cf(X)-f(Y+Z)+ cf(Y)}{X-Y}=0.
\end{eqnarray*}
Note that since $f(x)$ is not a monomial, $\mathcal{W}_c$ is actually a surface and not a curve.

The intersection $\mathcal{W}_c\cap H_{\infty}$ with the hyperplane at infinity is the curve 

\begin{eqnarray*}
\mathcal{C} &:& \frac{(X+1)^d- cX^d-(Y+1)^d+ cY^d}{(X-Y)}=0
\end{eqnarray*}
and by \cite[Theorems 4.1 and 4.2]{BC2021} whenever $c\notin N_{d-1}:=\{\xi^{i}(1-xi^{j})/(1-\xi^k) : i,j,k\in \{0,\ldots, d-2\}, k,j\neq 1\}$, where $\xi$ is a primitive $(d-1)$-th root of unity in $\overline{\mathbb{F}}_q$, $\mathcal{C}_c$ is nonsingular and therefore absolutely irreducible.

This shows that for any $c\notin N_{d-1}$ the surface $\mathcal{W}_c$ is absolutely irreducible. 

From now on we pick up $c\notin N_{d-1}$.  We want to bound the total number of singular points of $\mathcal{W}_c$.  First note that they are  only affine since $\mathcal{C}_c$ is nonsingular.

A point $P=(\alpha,\alpha,\gamma)$ is singular for $\mathcal{W}_c$ only if it is of multiplicity three for $\varphi(X,Y,Z):=f(X+Z)- cf(X)-f(Y+Z)+ cf(Y)=0$, since $P$ is of multiplicity one for the denominator $X-Y=0$. This happens only if   
\begin{eqnarray*}
\varphi(P)&=&\varphi^{(1,0,0)}(P)=\varphi^{(0,1,0)}(P)=\varphi^{(0,0,1)}(P)=\varphi^{(2,0,0)}(P)=\varphi^{(0,2,0)}(P)\\
&=&\varphi^{(1,1,0)}(P)=\varphi^{(0,1,1)}(P)=\varphi^{(1,0,1)}(P)=\varphi^{(0,0,2)}(P)=0.
\end{eqnarray*}

In particular 
$$\varphi^{(1,0,1)}(P)=f^{\prime\prime}(\alpha+\gamma)=0=f^{\prime\prime}(\alpha+\gamma)-cf^{\prime\prime}(\alpha)=\varphi^{(2,0,0)}(P)$$ 
and thus  $f^{\prime\prime}(\alpha+\gamma)=f^{\prime\prime}(\alpha)=0$ and there are at most $(d-2)^2$ possibilities for $(\alpha,\alpha,\gamma)$.

Consider now singular points  of $\mathcal{W}_c$ off $X-Y=0$. In particular, they satisfy  
$$\begin{cases}
f(X+Z)- cf(X)-f(Y+Z)+ cf(Y)=0\\
f^{\prime}(X+Z)- cf^{\prime}(X)=0\\
f^{\prime}(Y+Z)- cf^{\prime}(Y)=0.
\end{cases}$$

Such a system defines a variety $\mathcal{U}_c$ which is the intersection of three surfaces in $\mathbb{P}^3(\mathbb{F}_q)$ and $\mathcal{U}_c$ is of dimension $0$. To see this it is enough to observe that $\mathcal{U}_c\cap H_{\infty}$ is precisely the set of of singular points of $\mathcal{C}$ which is empty. This means that $\mathcal{U}_c$ cannot be of dimension larger than $0$ otherwise $\#(H_{\infty}\cap\mathcal{U}_c) $ would be positive. This shows that for any  $c\notin N_{d-1}$  the number of  singular points of $\mathcal{W}_c$ is $O(1)$. 

Now we bound the number of tangent planes to $\mathcal{W}_c$ of the type $\pi_z:Z=z$. Clearly, if $\pi_z$ is tangent at $P=(\alpha,\beta,z)\in \mathcal{W}_c$ then in particular  
$$\begin{cases}
f(\alpha+z)- cf(\alpha)-f(\beta+z)+ cf(\beta)=0\\
f^{\prime}(\alpha+z)- cf^{\prime}(\alpha)=0\\
f^{\prime}(\beta+z)- cf^{\prime}(\beta)=0,
\end{cases}$$
if $\alpha\neq \beta$ and 
$$f^{\prime\prime}(\alpha+z)-cf^{\prime\prime}(\alpha)=0=f^{\prime}(\alpha+z)-cf^{\prime}(\alpha)$$
if $\alpha=\beta$. In the former case we already saw that the number of solutions $(\alpha,\beta,z)$ is finite. In the latter case the two plane curves 
$$f^{\prime\prime}(X+Z)-cf^{\prime\prime}(Z)=0$$
and 
$$f^{\prime}(X+Z)-cf^{\prime}(Z)=0$$
do not share any component since their intersections with the line $\ell_{\infty}$ are disjoint and thus by B\'ezout's Theorem the number of intersection points, and thus the number of $(\alpha,\alpha,z)$ is $O(1)$. This shows that there is $O(1)$ of $z\in \overline{\mathbb{F}}_q$ such that  $\pi_z$ is tangent to $\mathcal{W}_c$. The claim now follows from Proposition \ref{Prop:Sections}.

\end{proof}

\begin{remark}
Note that in the lemma above the constant $B$ can be taken as $(d-2)^3$.
\end{remark}

As a byproduct of the lemma above we can easily deal with the case $_c \delta_F=1$.
\begin{corollary}
Let $q=p^n$, $p$ a prime, and $f\in \mathbb{F}_{q}[x]\setminus \mathbb{F}_{q}[x^p]$, with $d=\deg(f)$, $p\nmid d(d-1)$, and suppose that $f(x)$ is not a monomial. Then, the number of $c\in \mathbb{F}_q$ for which $f$ can be P$c$N  is bounded by an explicit constant $B$ independent of $q$.
\end{corollary}
\begin{proof}
Consider again the surface
\begin{eqnarray*}
\mathcal{W}_c &:& \frac{f(X+Z)- cf(X)-f(Y+Z)+ cf(Y)}{X-Y}=0.
\end{eqnarray*}
In the proof of Lemma \ref{Lemma} it has already been proved that $\mathcal{W}_c$ is absolutely irreducible whenever $c$ does not belong to a specific set of values of size at most $(d-2)^3$. By Lang-Weil Theorem the number of affine $\mathbb{F}_q$-rational points of $\mathcal{W}_c$ is lower-bounded by 
$$q^2-A(d)q^{3/2},$$
where $A(d)$ is an absolute constant depending only on the degree of $\mathcal{W}_c$. The set of parallel planes $\pi_a : Z-a=0$, $a\in \mathbb{F}_q$, partition the set of its affine $\mathbb{F}_q$-rational points and thus there exists at least an $\overline{a}\in \mathbb{F}_{q}^*$ such that $\#(\pi_{\overline{a}}\cap\mathcal{W}_c)\geq q-A(d)q^{1/2}$. This means that the curve 
\begin{eqnarray*}
\mathcal{C}_{c,\overline{a}} &:& \frac{f(X+a)- cf(X)-f(Y+a)+ cf(Y)}{X-Y}=0
\end{eqnarray*}
has at least $q-A(d)q^{1/2}$ affine points and thus, since the line $X-Y=0$ is not a component of $\mathcal{C}_{c,\overline{a}}$, $\mathcal{C}_{c,\overline{a}}$ possesses at least $q-A(d)q^{1/2}-d$ affine $\mathbb{F}_q$-rational points off $X-Y=0$ and thus $f(X+a)- cf(X)$ is not a permutation polynomial and $f$ is not P$c$N.

\end{proof}

\begin{remark}
In the corollary above the constant $B$ arises from Lang-Weil Theorem and thus it can be computed applying the results in \cite{CafureMatera,Ghorpade_Lachaud,Ghorpade_Lachaud2,LN1983,WSchmidt,Bombieri}. For instance, applying \cite[Theorem 7.1]{CafureMatera} one can see that $B$ can be chosen as $\max\{6.3 d^{13/3},(d-2)^2\}$. 
\end{remark}

Our main result of this section is the following.

\begin{theorem}[Main Result]\label{thm:main}
Let $q=p^n$, $p$ a prime, and $f\in \mathbb{F}_{q}[x]\setminus \mathbb{F}_q[x^p]$, $f$ not a monomial, with $d=\deg(f)$.
Suppose that one of the following holds:
\begin{enumerate}
    \item $p=2$, $d$ is odd, and both the Hasse derivatives $f^{\prime}$ and $f^{\prime\prime}$ do not vanish;
    \item $p>2$ and $d\not \equiv 0,1 \pmod{p}$.
\end{enumerate}
The number of $c\in \mathbb{F}_q$ for which $_c\delta_{F}<\deg(f)$ is bounded by an explicit constant $B$ independent of $q$.
\end{theorem}

\begin{remark}
\begin{itemize}\label{rem:secondhasse}
    \item As it will be clear from the proofs of the  ancillary results below, the constant $B$ in the statement of Theorem \ref{thm:main} is smaller than $4d^2$.
    \item The condition on the Hasse derivatives in Theorem~\ref{thm:main} is not very restrictive. Indeed, the polynomials that violate these conditions can be classified explicitly: The polynomials $f$ in $\mathbb F_{2^n}[x]$ such that the first Hasse derivative is zero live in $\mathbb F_{2^n}[x^2]$, and the ones for which the second Hasse derivative is zero are exactly the ones of the form 
\[x\left(\sum^s_{i=0}a_i x^{i}\right)^4+\left(\sum^s_{i=0}b_i x^{i}\right)^4.\]
    For instance, note that a sufficient condition for which both Hasse derivatives $f^{\prime}$ and $f^{\prime\prime}$ do not vanish is the existence of at least  one monomial in  $f(x)$ of degree $i\equiv 3\pmod{4}$.
    \item Theorem~\ref{thm:main} states that if $\deg(f)$ is small compared to the field size $q$, it is inevitable that there exist many $c\in \F_q$ such that $_c\delta_{F}=\deg(f)$, which is the maximal possible (and thus worst) $c$-uniformity. Note that $\deg(f)$ here refers to the degree of $f$ as a polynomial and not to the algebraic degree (which is in the characteristic 2 case equivalent to the highest binary weight of a monomial of $f$). This means that this result even applies to functions with high algebraic degree since it is clearly possible that a function with high algebraic degree has comparatively low degree as a polynomial.
\end{itemize}

\end{remark}

\noindent The proof of Theorem \ref{thm:main} involves the two surfaces 

\begin{eqnarray*}
\mathcal{W}_1 &:& \frac{f^{\prime}(X+Z)(f(Y)-f(X))-(f(Y+Z)-f(X+Z))f^{\prime}(X)}{(X-Y)Z}=0,\\
\mathcal{W}_2 &:& \frac{f^{\prime}(Y+Z)f^{\prime}(X) -f^{\prime}(X+Z)f^{\prime}(Y)}{(X-Y)Z}=0.\\
\end{eqnarray*}

Note that at this first step we strongly need that the polynomial $f$ is not a monomial, otherwise $\mathcal{W}_1$ and $\mathcal{W}_2$ are curves, being defined by a homogeneous polynomial in three variables, and our arguments do not apply.

\begin{proposition}\label{Prop:Surfaces}
Let $q=p^n$, $p\nmid d$. The two surfaces $\mathcal{W}_1$ and $\mathcal{W}_2$ do not share any component. 
\end{proposition}
\begin{proof}
Consider the two curves $\mathcal{C}_1 :=\mathcal{W}_1\cap H_{\infty}$ and $\mathcal{C}_2 :=\mathcal{W}_2\cap H_{\infty}$.
Then 
$$\mathcal{C}_1: \frac{(X+1)^{d-1}(Y^d-X^d)- X^{d-1}((Y+1)^d-(X+1)^d)}{X-Y}=0$$
and 
$$\mathcal{C}_2: \frac{(Y+1)^{d-1}X^{d-1}-(X+1)^{d-1}Y^{d-1}}{X-Y}=0.$$
It is readily seen that $\mathcal{C}_2$ factorizes as 
$$\prod_{\xi} \Big((Y+1)X-\xi (X+1)Y \Big)=0,$$
where $\xi$ runs over the set of the $d-1$-th roots of unity distinct from $1$. 

Let $\ell_\xi : (Y+1)X-\xi (X+1)Y=0$ be one of the components of $\mathcal{C}_2$. In order to show that $\ell_\xi \not\subset \mathcal{C}_2$, consider 
$$
\begin{cases}
(Y+1)X-\xi (X+1)Y=0\\
(X+1)^{d-1}(Y^d-X^d)- X^{d-1}((Y+1)^d-(X+1)^d)=0.\\
\end{cases}
$$

Since $(Y+1)X-\xi (X+1)Y=0$, $(Y+1)^{d-1}X^{d-1}= (X+1)^{d-1}Y^{d-1}$ and thus $(Y+1)^{d}X^{d-1}= (X+1)^{d-1}Y^{d-1}(Y+1)$. So 
\begin{eqnarray*}
(X+1)^{d-1}(Y^d-X^d)- X^{d-1}((Y+1)^d-(X+1)^d)=\\
(X+1)^{d-1}(Y^d-X^d)-(X+1)^{d-1}Y^{d-1}(Y+1)  +X^{d-1}(X+1)^d=\\
(X+1)^{d-1}\Big((Y^d-X^d)-Y^{d-1}(Y+1)  +X^{d-1}(X+1)\Big)=\\
(X+1)^{d-1}\Big(X^{d-1}-Y^{d-1}\Big)\not\equiv 0.\\
\end{eqnarray*}
Since $\mathcal{C}_1$ and $\mathcal{C}_2$ do not share any component, so do the surfaces $\mathcal{W}_1$ and $\mathcal{W}_2$. 

\end{proof}

\begin{proposition}\label{Prop:Only2}
Let $q>(2d-2)(2d-3)+1$, $q=p^n$,  $f(x) \in \mathbb{F}_q[x]\setminus \mathbb{F}_q[x^p]$ not a monomial,  $p\nmid d=\deg(f)$. There exists a set $\Theta_c\subset \mathbb{F}_q$ of size at most $(2d-2)(2d-3)$ such that for any $c \in \mathbb{F}_q\setminus \Theta_c$ there exists at least an $a_c\in \mathbb{F}_q^*$ such that 
$f(x+a_c)-cf(x)=t$ has at most one multiple root in $\overline{\mathbb{F}}_q$ for any fixed $t\in \overline{\mathbb{F}}_q$.
\end{proposition}

\begin{proof}
 
Consider the system 
$$
\begin{cases}
f(x_1+a)-cf(x_1)=t\\
f(x_2+a)-cf(x_2)=t\\
f^{\prime}(x_1+a)-cf^{\prime}(x_1)=0\\
f^{\prime}(x_2+a)-cf^{\prime}(x_2)=0\\
a(x_1- x_2)\neq 0.
\end{cases}
$$
The solutions $(\overline{x_1},\overline{x_2},\overline{a},\overline{c},\overline{t})$ of this system correspond to values $\overline{a},\overline{c},\overline{t}$  for which there exist two multiple roots of $f(x+\overline{a})-\overline{c}f(x)=\overline{t}$.

The above system is equivalent to 
\begin{equation}\label{Eq:Sistema}
\begin{cases}
f(x_1+a)-cf(x_1)=t\\
\frac{f(x_2+a)-f(x_1+a)}{f(x_2)-f(x_1)}=c\\
a(x_1- x_2)\neq 0\\
\frac{f^{\prime}(x_1+a)(f(x_2)-f(x_1))-(f(x_2+a)-f(x_1+a))f^{\prime}(x_1)}{a(x_1-x_2)}=0\\
\frac{f^{\prime}(x_2+a)f^{\prime}(x_1) -f^{\prime}(x_1+a)f^{\prime}(x_2)}{a(x_1-x_2)} =0.\\
\end{cases}
\end{equation}

The last two equations define the surfaces  $\mathcal{W}_1$ and $\mathcal{W}_2$ considered in Proposition \ref{Prop:Surfaces}. Since such surfaces do not share any component, their intersection is of dimension one (i.e. union of curves) and of degree at most $(2d-2)(2d-3)$. 

Thus, apart from a small (at most $(2d-2)(2d-3)$) number of $a$, the intersection $W_1\cap W_2\cap (Z=a)$ consists of at most $(2d-2)(2d-3)$ points (on the algebraic closure $\overline{\mathbb{F}}_q$). Since $q>(2d-2)(2d-3)+1$, there exists $\overline{a}\in \mathbb{F}_q^*$ for which the total number of solutions $(x_1,x_2,c,t)$ of System \eqref{Eq:Sistema} is upped bounded by $(2d-2)(2d-3)$. Let $\Theta$ be the set of all the solutions of System \eqref{Eq:Sistema} for such an $\overline{a}$ and consider   $\Theta_c:=\{c \in \mathbb{F}_q : \exists (\overline{x_1},\overline{x_2},c,\overline{t}) \in \Theta\}$. Clearly $\#\Theta_c\leq\#\Theta \leq(2d-2)(2d-3)$. Thus, for any $c \in \mathbb{F}_q\setminus \Theta_c$ we have that 
$f(x+\overline{a})-cf(x)=t$ has at most one multiple root $x_1$ for any fixed $t\in \overline{\mathbb{F}}_q$ and the claim follows.

\end{proof}

\begin{proposition}\label{Prop:No3}
Let $q>\max\{2d-4,(d-1)^2\}$, $q=p^n$,  $f(x) \in \mathbb{F}_q[x]\setminus \mathbb{F}_q[x^p]$ not a monomial, $d=\deg(f)$, and $p\nmid d(d-1)$. There exists a set $\Theta^{\prime}_c\subset \mathbb{F}_q$ of size at most $2d-4$ such that for any $c \in \mathbb{F}_q\setminus \Theta^{\prime}_c$ there exists at least an $a_c\in \mathbb{F}_q^*$ such that 
$f(x+a_c)-cf(x)=t$ has no root in $\overline{\mathbb F}_q$ of multiplicity larger than $2$ for any fixed $t\in \overline{\mathbb{F}}_q$.
\end{proposition}
\begin{proof}
Consider the system 
$$
\begin{cases}
f(x+a)-cf(x)=t\\
f^{\prime}(x+a)-cf^{\prime}(x)=0\\
f^{\prime\prime}(x+a)-cf^{\prime\prime}(x)=0\\
a\neq 0.
\end{cases}
$$
Any solution $(\overline{x},\overline{a},\overline{t},\overline{c})$ provides values $\overline{a},\overline{t},\overline{c}$ such that 
$f(x+\overline{a})-\overline{c}f(x)=\overline{t}$ has a root of multiplicity at least three. 

The system above is equivalent to 
\begin{equation}\label{Sist:2}
\begin{cases}
f(x+a)-cf(x)=t\\
f^{\prime}(x+a)-cf^{\prime}(x)=0\\
a\neq 0\\
\frac{f^{\prime}(x)f^{\prime\prime}(x+a)-f^{\prime}(x+a)f^{\prime\prime}(x)}{a}=0.\\
\end{cases}
\end{equation}

The last equation in System \eqref{Sist:2} is a non-vanishing equation of degree $2d-4$ for any $a$. To see this, let $f(x)=x^d+\alpha x^{d-1}+\cdots $ and thus  
\begin{eqnarray*}
f^{\prime}(x)&=&d x^{d-1}+\alpha(d-1) x^{d-2}+\cdots,\\
f^{\prime\prime }(x)&=&d (d-1)x^{d-2}+\alpha(d-1)(d-2) x^{d-3}+\cdots,
\end{eqnarray*}
and the leading coefficient  of 
$$f^{\prime}(x)f^{\prime\prime}(x+a)-f^{\prime}(x+a)f^{\prime\prime}(x) $$ 
is $-ad^2(d-1)\neq 0$.

Note that $f^{\prime}(x+a)-cf^{\prime}(x)$ and $f^{\prime\prime}(x)$ are non-vanishing polynomials. The number of $a$ for which $f^{\prime}(x+a)$ and $f^{\prime}(x)$ share a factor is upperbounded by $(\deg(f^{\prime}))^2\leq (d-1)^2$. Let $\overline{a}\in \mathbb{F}_q^*$ be such that $f^{\prime}(x+a)$ and $f^{\prime}(x)$ do not share any factor.

For this fixed $\overline{a}$, System \eqref{Sist:2} admits at most $2d-4$ solutions $(\overline{x},\overline{a},\overline{t},\overline{c})$ and the claim follows.

\end{proof}

\begin{proposition}\label{Prop:No3_2}
Let $q=2^h>\max\{2d-4,(d-1)^2\}$ and $f(x) \in \mathbb{F}_q[x]\setminus \mathbb{F}_q[x^2]$ not a monomial. 
Suppose that both the Hasse derivatives $f^{\prime}$ and $f^{\prime\prime}$ do not vanish.
There exists a set $\Theta^{\prime}_c\subset \mathbb{F}_q$ of size at most $2d-4$ such that for any $c \in \mathbb{F}_q\setminus \Theta^{\prime}_c$ there exists at least an $a_c\in \mathbb{F}_q^*$ such that 
$f(x+a_c)-cf(x)=t$ has no root in $\overline{\mathbb F}_q$ of multiplicity larger than $2$ for any fixed $t\in \overline{\mathbb F}_q$.
\end{proposition}
\begin{proof}
 The argument is the same as in the proof of Proposition \ref{Prop:No3}. Now 
 \begin{eqnarray*}
 f^{\prime}(x)&=&\alpha_r x^r+\alpha x^{r-2}+\cdots,\\
 f^{\prime\prime}(x)&=&\alpha_s x^s+\beta x^{s-1}+\cdots,
 \end{eqnarray*}
 where $r+1$ and $s+2$ are the largest degrees not divisible by $2$ and  equivalent to $3\pmod{4}$, respectively. 
So, the  leading coefficient of
$$f^{\prime}(x)f^{\prime\prime}(x+a)-f^{\prime}(x+a)f^{\prime\prime}(x) $$ 
is $a\alpha_s\alpha_r\neq 0$ and the claim follows.

\end{proof}

Combining Propositions \ref{Prop:Only2},  \ref{Prop:No3}, and \ref{Prop:No3_2} we have the following.
\begin{proposition}\label{Prop:finale}
Let $q>(2d-1)(2d-3)$, $q=p^n$, and $f(x) \in \mathbb{F}_q[x]\setminus \mathbb{F}_q[x^p]$ not a monomial, $d=\deg(f)$, with 
\begin{enumerate}
    \item $p\nmid d(d-1)$ if $p$ is odd;
    \item the Hasse derivatives $f^{\prime}$ and $f^{\prime\prime}$ do not vanish if $p=2$.
\end{enumerate}
There exists a set $\Psi_c\subset \mathbb{F}_q$ of size at most $(2d-1)(2d-3)$ such that for any $c \in \mathbb{F}_q\setminus \Psi_c$ there exists at least an $a_c\in \mathbb{F}_q^*$ such that 
$f(x+a_c)-cf(x)=t$ has either all distinct roots or precisely one double root in $\overline{\mathbb{F}}_q$ for any fixed $t\in \overline{\mathbb{F}}_q$.
\end{proposition}

We are now ready to prove the main theorem of this section.

\begin{proof}[Proof of Theorem \ref{thm:main}]
Using Corollary \ref{cor:geomtoarith}  it is enough to prove that the number of pairs $(a,c)\in \mathbb{F}_q^*\times \mathbb F_q$ such that the geometric Galois group of $F=f(x+a)- cf(x)-t\in \mathbb F_q(t)[X]$ is not the symmetric group $\mathcal S_n$, is $O(1)$.

Observe that thanks to Lemma \ref{lemma:orbits} combined with Proposition \ref{Prop:finale} we have that the inertia groups of the Galois group of $F$ are all isomorphic to $C_2$ and generated by single transpositions. Since $G$ is generated by its inertia groups thanks to Lemma \ref{lemma:genbyinertia},   $G=\textrm{Gal}(F\mid \overline{\mathbb F}_q)$ is a transitive subgroup of $S_n$ generated by transpositions (since $F$ is irreducible for any fixed pair $(a,c)$). 
Proposition \ref{prop:gensn} now implies that $G=S_n$.
Therefore we obtain that $G=\textrm{Gal}(F\mid \mathbb F_q)=S_n$ as well, from which the claim follows thanks to Corollary \ref{cor:geomtoarith}.

%
%


\end{proof}

\section{The feasibility of differential attacks based on the c-differential uniformity} \label{sc:2}

The focus in the research on $c$-differential uniformity has so far been almost purely on determining the $c$-differential uniformity of specific functions. The actual use case has been mostly neglected, which is surprising given that a clear attack based on bad $c$-differential uniformities has not been presented yet. 

While the "standard" differential uniformity measures the probability of a difference propagating through the S-box (indeed $DDT_F[a,b]$ is indeed the probability of a difference $a$ turning into a difference $b$ times $2^n$), it is not immediately clear what kind of statistical bias the $c$-differential uniformity measures. Clearly, the distribution of the values of $F(x+a)-cF(x)$ is also a measure of differential bias (since it compares the inputs of $x$ and $x+a$), but the output is for $c \neq 1$ not a usual difference itself, so the construction of a differential trail that tracks the propagation through several rounds of the cipher is not readily possible. The $c$-differential uniformity thus does not measure a propagation of usual differences. 

Unlike the multiplicative differentials~\cite{mult} mentioned as inspiration for the $c$-differential uniformity in~\cite{cdiff}, the $c$-differential uniformity also does not measure the propagation of multiplicative "differences" of the form $(x,\alpha x)$ through the cipher since, as mentioned, the input difference used is actually the regular addition. It is however possible to find a kind of "difference" such that the $c$-differential uniformity does measure the propagation of this "difference", as we present now. 

\subsection{A general differential attack}
Instead of using the usual difference $a-b$ (which, in the case of the usual $\F_2^n$ setting, boils down to the XOR $a \oplus b$), other binary operation can of course be used, for instance in the case of the multiplicative differentials from~\cite{mult} mentioned above, this is the modular multiplication. So let $\circ \colon \F_p^n \times \F_p^n \rightarrow \F_p^n$ be a binary operation. We can then consider the propagation of pairs of those generalized differences $(x,x \circ a)$ and, identical to the usual differential attack, we can attempt to use biases in the probabilities of the propagation of those differences. 

Generalizing the usual differential uniformity of a function $F \colon \F_{p^n} \rightarrow \F_{p^n}$, we can define the $\circ$-differential uniformity and the $\circ$-DDT. 
\begin{definition} \label{def:general}
	Let $F \colon \F_{p^n} \rightarrow \F_{p^n}$ be a function. We define for all $a,b \in \F_{p^n}$
	\[_\circ\DDT_F[a,b]=\#\{x \in \F_{p^n}\colon F(x \circ a)=b \circ F(x)\}\]
	and 
	\[_\circ\delta_F={\max_{x \circ a \neq x, b \in {\F_{p^n}}}} {_\circ\DDT_F[a,b]}.\]
\end{definition}
Clearly, for $\circ=+$ one recovers the usual DDT and differential uniformity. But, more interestingly, this framework actually also allows us to recover the $c$-differential uniformity. Indeed, setting $a \circ_c b := a+cb$ for all $a,b \in \F_{p^n}$ and a fixed $c \in \F_{p^n}$, we get
\[F(x \circ_c a)=b \circ_c F(x) \Leftrightarrow F(x+ca)=b+cF(x)\Leftrightarrow F(x+ca)-cF(x)=b.\]
Since $c$ is fixed, a simple transformation $a \mapsto a/c$ then immediately relates the $c$-DDT with the $\circ$-DDT for this choice of $\circ$, and the respective differential uniformities are identical. In this sense, the $c$-differential uniformity is just a new special case of differential uniformity for this specific choice of the binary operation $\circ$. The $c$-differential uniformity thus seems to be just a tool to measure the resistance of a cipher against a specific differential attack based on this operation. We want to note that the general form of differential as described in Definition~\ref{def:general} was analysed in a series of papers~\cite{genop,genop2} with the idea to find specific binary operations that can lead to efficient differential attacks (or, possibly, for a malicious designer, to a "hidden", non-public binary operation that serves as a trapdoor to a cipher resistant against the usual attacks). However, the authors in those papers only analyse a subclass of binary operations that excludes the specific binary operation $\circ$ that we identified as being related to the $c$-differential uniformity here. 

\subsection{The differential attack based on weak $c$-differential uniformity}

Let us now consider a potential differential attack using this specific binary operation $\circ_c$ defined via $a\circ_c b = a+cb$ for all $a,b \in \F_{p^n}$ and $c \in \F_{p^n}^*$. As explained briefly in the introduction on the classic differential attack earlier,  the differential attack is possible for many block ciphers since they use simple key addition as a primitive operation, which means that differences propagate in the same way regardless of the key. Moreover, differences propagate also unchanged through the linear layer. We now check if/when those two properties are satisfied by $\circ_c$.

We start with a consideration of the linear layer. The following result states that a linear layer applied to the input of the function does not change its $c$-differential uniformity, however a linear layer applied to the output generally does.
\begin{theorem} \label{thm:linear}
	Let $F\colon \F_{p^n} \rightarrow \F_{p^n}$ be a function and $A \colon \F_{p^n} \rightarrow \F_{p^n}$ be an $\F_p$-affine permutation, where $A=L+s$ and $L$ is the linear part of $A$. Then 
	\[_c\DDT_F[a,b]={_c\DDT_{F\circ A}[L(a),b]}\]
	and in particular
	\[_c\delta_{F} = {_c\delta_{F\circ A}}.\]
	
	Moreover, if $A$ is affine over $\F_{p^l}$ where $l=[\F_p(c)\colon \F_p]$, then 
	\[_c\DDT_F[a,b]= {_c\DDT_{A\circ F}[a,L^{-1}(b-(1-c)s)]}\] and 
	\[_c\delta_{F} = {_c\delta_{A\circ F}}.\]
	However, generally $_c\delta_{F} \neq {_c\delta_{A\circ F}}$ if $c\notin \F_p$.
\end{theorem} 
\begin{proof}
	Clearly, $(F\circ A)(x+a)-c(F\circ A)(x)=b$ if and only if $F(L(x)+L(a)+s)-cF(L(x)+s)=b$, and the first results follows since $L$ is a permutation.\\
	
	On the other hand, $(A\circ F)(x+a)-c(A\circ F)(x)=b$ if and only if 
	\[F(x+a)+L^{-1}(s)-L^{-1}(cL(F(x))+cs)=F(x+a)-L^{-1}(cL(F(x)))+L^{-1}((1-c)s)=L^{-1}(b).\]
	If $L(cx)=cL(x)$ for all $x$, then $L^{-1}(cL(x))=cx$ and $_c\DDT_F[a,b]={_c\DDT_{A\circ F}[a,L^{-1}(b-(1-c)s)]}$. We check when $L(cx)=cL(x)$  occurs. Writing $L$ as a polynomial $L=\sum_{i=0}^{n-1} a_ix^{p^i}$, we get
	\begin{align*}
		L(cx)&=\sum_{i=0}^{n-1} c^{p^i}a_ix^{p^i}\\
		cL(x)&=\sum_{i=0}^{n-1} ca_ix^{p^i}.
	\end{align*}
	Comparing coefficients yields $c^{p^i}=c$ for all $i$ with $a_i\neq 0$. $c^{p^i}=c$ is equivalent to $c \in \F_{p^i}$ or $i=0$, so we conclude that $L(cx)=cL(x)$ for all $x \in \F_{p^n}$ if and only if $a_i=0$ unless $i=0$ or $c \in \F_{p^i}$ which implies that $a_i=0$ unless $i=0$ or $l|i$. This shows that $L$ is linear over $\F_{p^l}$. \\
	
	If this condition does not hold, it is easy to construct examples by computer search that show $_c\delta_{F} \neq _c\delta_{A\circ F}$.

\end{proof}

Theorem~\ref{thm:linear} shows that unless one picks very specific linear layers or $c \in \F_p$ (which in the $\charac 2$ case most interesting for applications only holds for $c=1$, i.e. the classical differential uniformity), the $c$-differential uniformity is actually affected by the choice of the linear layer. In particular, unlike the classical differential attack, the resistance of the cipher cannot be broken down to the S-box level, since linear layers used in block ciphers are of course generally not linear over extension fields. \\

Let us now consider the interaction of the $\circ$-differences with the process of key addition. For the classical differences, the key addition does not impact any differences since 
\begin{equation}
(x+\Delta)+k-(x+k)=\Delta
\label{eq:differences_classical}
\end{equation}

 for any choice of $x,k,\Delta$. For our binary operation $\circ_c$, let $\overline{\circ_c} \colon \F_{p^n} \times \F_{p^n} \rightarrow \F_{p^n}$, defined as $a \overline{\circ_c} b := a-cb$, be the "inverse" of $\circ_c$ in the sense that $(a \circ_c b) \overline{\circ_c} b=a$ for all $a,b \in \F_{p^n}$. We now consider the differences with respect to $\circ_c$ by substituting $+$ and $-$ of the classical differences in Eq.~\eqref{eq:differences_classical} with $\circ_c$ and $\overline{\circ_c}$ (while of course keeping the regular addition for the key addition). This yields
\[((x\circ_c \Delta) + k)\overline{\circ_c}(x+k)=x+c\Delta+k-c(x+k)=c\Delta-(c-1)(x+k).\]
It is immediate to see that the differences now are neither independent from the subkey $k$ nor from the message $x$ if $c \neq 1$. \\

So, in closing, it seems that a differential attack based on $\circ_c$-differences (attempting to abuse high $c$-differential uniformity) has several practical challenges as both the linear layers and the key addition process used in the vast majority of block ciphers make such an attack considerably harder than an attack relying on the classical concept of differences. For block ciphers that do not use a simple key addition but possibly another primitive operation, differential attacks based on different binary operations as lined out in this section might be of practical interest. \\

Regardless, the $c$-differential uniformities still measure biases in the distribution of differences and it might still theoretically be possible to construct an attack different than the one considered here to abuse this bias. However, it seems to be clear that an analysis would be made considerably more difficult by the fact that linear layers play a more significant role than in the classical differential attack and its derivatives.


\section{Conclusions and open problems}

Our main contribution concerns the $c$-differential uniformity of polynomials and states that for a generic polynomial  there exists only a thin set of instances of  $c\in \mathbb{F}_q$ for which  $_c\delta_F$ is not the worst possible.   
Our investigation involves techniques from both Algebraic Geometry in positive characteristic and Galois Theory and tells us that in order to avoid the differential biases encoded in the $c$-differential uniformity, polynomials need to respect specific constraints on their degree structure. 

More generally, we show that extreme statistical differential biases for a big class of functions are inevitable. While our analysis in Section~\ref{sc:2} indicates that constructing attacks based on those biases is not easy, it remains open if other attacks exploiting $c$-differential uniformities are possible. 

An obvious question is if it is possible to extend the techniques we used in this paper to analyze other statistical biases, in particular it would be interesting to see if results of the form of Theorem~\ref{thm:main} can be achieved.

An interesting theoretical open question concerns the possibility to obtain similar results involving weaker constraints, for instance dropping the condition $p\nmid \deg(f)$ in Theorem~\ref{thm:main}. 

In Section~\ref{sc:2}, we discussed a \emph{general} differential attack with an arbitrary binary operation replacing the XOR. In~\cite{genop}, the authors investigate a similar generalized attack (starting from a different motivation), and argue that for some specific binary operations, there are enough weak keys that allow the exploitation of certain biases. While the results are not applicable to the case of $c$-differential uniformities, it would be interesting if the operations investigated in~\cite{genop} behave similarly to  the $c$-differential uniformity as described in this paper.

\section{Acknowledgments}
\noindent 
The first author is supported by the Italian National Group for Algebraic and Geometric Structures and their Applications (GNSAGA - INdAM).
The second and third authors are supported by NSF grant 2127742.

 \bibliographystyle{IEEEtran}
 \bibliography{bib}

\begin{IEEEbiographynophoto}{Daniele Bartoli}

 was born in Molfetta, Italy, on September 22, 1985. He
received the degree in Mathematics from the University of Perugia, Italy,
in September 2008 and the Ph.D. degree in Mathematics and Computer
Science from the same University in 2012, with a dissertation from coding
theory. Currently, he is Associate Professor at the Department of
Mathematics and Computer Science, University of Perugia, Perugia, Italy.
His research interests are in coding theory, including quantum coding, and
Galois geometries.

\end{IEEEbiographynophoto}


\begin{IEEEbiographynophoto}{Lukas K\"olsch}
 reveived the Masters Degree from Otto-von-Guericke University Magdeburg in 2017 and the Ph.D. degree
in 2020 from University of Rostock, Germany under the supervision of Gohar Kyureghyan. He is currently a postdoctoral scholar at University of South Florida. He is interested in Boolean functions, symmetric cryptography, and Galois geometry.
\end{IEEEbiographynophoto}

\begin{IEEEbiographynophoto}{Giacomo Micheli}
 graduated at the University of Rome “La Sapienza” in July 2012. He completed his
Ph.D. with distinction at the Zurich Graduate School in Mathematics in October 2015 under the supervision
of Prof. Joachim Rosenthal. He is currently a tenure-track assistant professor at the University of South
Florida and Co-Director of the Center for Cryptographic Research at USF.
\end{IEEEbiographynophoto}

\end{document}